\documentclass[twocolumn]{autart}    % Enable this line and disable the
\usepackage{amsmath,bm}
\usepackage{color}
\usepackage{multirow}
\usepackage{mathrsfs}
\usepackage{tcolorbox}
\usepackage{subfigure}
\usepackage{dsfont}
\usepackage{amssymb}
\usepackage{setspace}
\usepackage{graphicx}
\usepackage{float}
\usepackage{stfloats}
\usepackage{graphicx}
\usepackage[export]{adjustbox}
\newtheorem{remark}{Remark}
\newtheorem{lemma}{Lemma}

\newtheorem{assumption}{Assumption}

\newtheorem{proposition}{Proposition}

\newtheorem{proof}{\textbf{Proof}}
\def\begcen{\begin{center}}
\def\endcen{\end{center}}

\newcommand{\col}{\mbox{col}}

\def\diag{\mbox{diag}}

\def\bw{{\bf w}}
\def\bx{{\bf x}}
\def\bv{{\bf v}}
\def\bi{{\bf i}}
\def\bc{{\bf c}}

\def\calm{{\mathcal M}}

\def\call{{\mathcal L}}

\def\hal{\frac{1}{2}}

\def\L2e{{\mathcal L}_{2e}}

\def\rea{\mathbb{R}}

\def\diag{\mbox{diag}}

\def\et{\epsilon_t}

\def\l2{{\mathcal L}_2}
\def\l2e{{\cal L}_{2e}}

\def\rea{\mathbb{R}}

\def\diag{\mbox{diag}}

\def\begequarr{\begin{eqnarray}}
\def\endequarr{\end{eqnarray}}
\def\begequarrs{\begin{eqnarray*}}
\def\endequarrs{\end{eqnarray*}}
\def\begarr{\begin{array}}
\def\endarr{\end{array}}
\def\begequ{\begin{equation}}
\def\endequ{\end{equation}}
\def\lab{\label}
\def\begdes{\begin{description}}
\def\enddes{\end{description}}
\def\begenu{\begin{enumerate}}
\def\begite{\begin{itemize}}
\def\endite{\end{itemize}}
\def\endenu{\end{enumerate}}

\def\lef[{\left[\begin{array}}
\def\rig]{\end{array}\right]}
\def\qed{\hfill$\Box$}
\def\begcen{\begin{center}}
\def\endcen{\end{center}}
\def\begrem{\begin{remark}\rm}
\def\endrem{\end{remark}}
\def\TAC{{\it IEEE Trans. Autom. Control}}

\def\AUT{{\it Automatica}}

\def\ACC{{\it American Control Conf.}}

\def\begmat#1{\begin{bmatrix}#1\end{bmatrix}}
\def\begali#1{\begin{align}{#1}\end{align}}
\def\begalis#1{\begin{align*}{#1}\end{align*}}

\usepackage{color}

\begin{document}

\begin{frontmatter}

\title{A High Performance Globally Exponentially Stable Sensorless Observer for the IPMSM: Theoretical and Experimental Results} 
\thanks[footnoteinfo]{An abridged version \cite{IFAC} of this paper will be presented at the 22nd IFAC World Congress, 2023. This work was partially supported by the National Research Foundation of Korea (NRF) grant funded by the Korean government (MSIT) (NRF-2021R1I1A3059676). Corresponding author: Bowen Yi. }
    % Add the
\author[USYD]{Bowen Yi}\ead{bowen.yi@sydney.edu.au}, 
\author[ITAM]{Romeo Ortega}\ead{romeo.ortega@itam.mx},
              % e-mail address % (ead) as shown
\author[HU]{Jongwon Choi}\ead{jongwon@hnu.kr},
\author[POSTECH]{Kwanghee Nam}\ead{kwnam@postech.ac.kr}

\address[USYD]{Australian Centre for Field Robotics, The University of Sydney, Sydney, NSW 2006, Australia}
            % full addresses
\address[ITAM]{Departamento Acad\'{e}mico de Sistemas Digitales, ITAM, Progreso Tizap\'an 1, Ciudad de M\'exico, 04100, M\'{e}xico}
%\vspace{-0.2cm}
%\address[ITMO]{Department of Control Systems and Informatics, ITMO University, Saint Petersburg 197101, Russia}
%\vspace{-0.2cm}
%\vspace{-0.2cm}
\address[HU]{Department of Electrical and Electronic Engineering, Hannam University, Daejeon 34430, Korea}
\address[POSTECH]{Department of Electrical Engineering, Pohang University of Science and Technology (POSTECH), Pohang 790-784, Korea}
%\vspace{-0.2cm}

                             % keyword list or with the

\begin{abstract}  
In a recent paper \cite{ORTetal} the authors proposed the first solution to the problem of designing a {\em globally exponentially stable} (GES) flux observer for the interior permanent magnet synchronous motor. However, the establishment of the stability proof relies on the assumption that the adaptation gain is sufficiently {\em small}---a condition that may degrade the quality of the transient behavior. In this paper we propose a new GES flux observer that overcomes this limitation ensuring a high performance behavior. The design relies on the use of a novel theoretical tool---the generation of a {\em ``virtual" invariant manifold}---that allows the use of the more advanced Kreisselmeier's regression extension estimator, instead of a simple gradient descent one. We illustrate its superior transient behavior via extensive simulations and {\em experiments}. 
\end{abstract}
\end{frontmatter}

%
%%%%%%%%
\section{Introduction}
\label{sec1}
%%%%%%%%
%

%%%%%%%%
%

%\rom{To be done}

Sensorless control is one of the most important topics in the field of motor drives, for which we assume only access to electrical coordinates, i.e., the stator currents and voltages. It is of paramount importance in industry since sensorless control of motors can reduce hardware complexity, improve reliability, and achieve cost reduction \cite{HOL,NAM}. Meanwhile, this problem has received particular attention from the automatic control community due to its theoretical interests, i.e., the dynamical model of the permanent magnet synchronous motor (PMSM) does not belong to any observable canonical forms that have been studied comprehensively \cite{BER}. To be precise, the main constructive challenges arise from the highly nonlinear output functions.

It is widely recognized that the centrality of sensorless control is the high-performance algorithm to estimate angular positions. The existing estimation algorithms can be broadly classified into two classes, which make use of high-frequency and fundamental components in stator currents, respectively. The former is known as the signal injection approach, which is applicable to the low speed region; and the latter is more suitable at middle and high speeds \cite{NAM}. In this paper, we are limited to the (fundamental frequency) model-based approaches. Indeed, both the back-emf and flux contain the information of angular positions implicitly, thus providing distinct technical routes to reconstruct mechanical coordinates. A salient feature of flux is its \emph{constant} magnitude regardless of rotating speeds, and this fact triggers intensive research activity on \emph{flux observers} for PMSMs \cite{BOBetal,BOLetal,ORTetalTCST} in both the power electronics and control communities.

An important contribution to motor flux observers was reported in \cite{ORTetalTCST}, where it was observed that the motor dynamics satisfies an algebraic relation that can be used to formulate an optimization criterion whose gradient can be explicitly computed for its use in a gradient descent-based observer. This observation led to the design of the first locally stable observer whose practical viability was established in \cite{LEEetal}. Later on, this design was rendered global (via convexification) in \cite{MALetal} and then combined with an adaptive mechanism in \cite{BERPRA}, both achieving global asymptotic stability (GAS). Following the research line initiated in \cite{POUPRAORT}, a Kazantzis-Kravaris-Luenberger observer was recently proposed in \cite{BERPRAtac}, which is capable to deal with the scenario of unknown resistance. Note that above-mentioned methods are tailored for the \emph{surface-mounted} PMSM. However, in recent years the \emph{interior} permanent magnet synchronous motor (IPMSM) has become widely popular in a variety of industrial applications due to its high power density, cheaper cost and reluctance torque. The dynamical model of the IPMSM is far more complicated than the one of the surface-mounted PMSM. In particular, the equivalent of the algebraic relation mentioned above, which was exploited in \cite{ORTetalTCST}, depends now on the rotor angle---hence been unsuitable for its use in a gradient descent based observer. For this reason the design of a globally stable flux observer for IPMSMs was considered to be a wide open problem in \cite[Section VI]{ORTetalTCST}, a fact also endorsed in \cite[Section VII]{BERPRAtac}.

This problem has been recently solved by the authors in \cite{ORTetal}. A key observation to establish this result was to formulate the estimation objective in terms of the active, instead of the magnetic, flux. A similar shift of the problem formulation was (unknowingly) pursued in \cite{MALPRA}, leading also to a globally stable design. It is also worth noting the work of \cite{VERetal} where local convergence is reported for a scheme incorporating parameter adaptation. As of this date, the observer design reported in the literature with the strongest stability property, namely GES,  is the one in \cite{ORTetal}. Its basic idea is to obtain a {\em nonlinear} regression equation on the active flux---which can be viewed as a perturbed linear regressor---by adopting the filtering technique in \cite{CHOetal}. A potential drawback of this scheme is that to prove stability, we require a technical condition---selecting the adaptation gain {\em sufficiently small}. Such low adaptation gain design may limit the estimation performance, particularly in the transient stage. The aim of this paper is to propose an alternative solution, where we do not impose a restriction on the size of the adaptation gain, addressing in this way the issue of (potentially poor) transient behaviors.

In order to be able to design our new globally exponentially convergent flux observer, we add to the filtering technique of \cite{CHONAM,ORTetal}, a novel mathematical concept: the creation of a {\em virtual invariant manifold}. Thanks to this new technique we are able to replace the gradient-based observer design by the high performance Kreisselmeier's regression extension (KRE)-based estimator, reported in \cite{KRE} for the design of adaptive observers for linear time-invariant (LTI)  systems. We notice that the KRE-based estimator was  recently used for electromechanical systems in \cite{Li}---see \cite{ARAetal} for a detailed analysis of the properties of this estimator. To evaluate performance of the proposed sensorless observer for IPMSMs we present numerous simulation and {\em experimental} results.

\begin{figure}
\fbox{ \parbox { .92\linewidth}
{
\begin{center}
 {\bf Nomenclature}
\end{center}
\vspace{0.2cm}
  \renewcommand\arraystretch{1.4}
\small
\begin{tabular}{ll}
$\alpha\beta$ & Stationary axis reference frame quantities\\
$\bv,\bi \in \rea^2$ & Stator voltage and current [V, A] \\
$\lambda \in \rea^2$ & Stator flux [Wb]\\
$\bx \in \rea^2$ & Active flux [Wb] \\
$\theta \in {\mathbb S}$ & Rotor flux angle [rad] \\
$R$    & Stator winding resistance [$\Omega$]\\
$\psi_m$ & PM flux linkage constant  \\
$L_d,L_q$ & $d$ and $q$-axis inductances [H] \\
$L_0$ & Inductance difference $L_0:= L_d - L_q$ [H] \\
$L_s$ & {Averaged inductance} $L_s:= {L_d + L_q \over 2}$ [H] \\
%$|\cdot|$ & Euclidean norm of a vector \\
%$\|\cdot\|_\infty$ & $\call_\infty$ norm of  a signal \\
$p$ & Differential operator $p:= {d \over dt}$\\
$G(p)[w]$ & Action of $G(p) \in \rea(p)$ on a signal $w(t)$\\
\end{tabular}
}}
\end{figure}

%
%%%%%%%%
\section{Motor Model and Observer Problem Formulation}
\label{sec2}
%%%%%%%%
%
According to Faraday's Law, the electrical dynamics (in the stationary $\alpha\beta$ frame) is given by
\begequ
\label{ipmsm1}
\begin{aligned}
  \dot{\lambda} & = - R\bi+ \bv.
\end{aligned}
\endequ
For IPMSMs, the measurable current satisfies
\begin{equation}
\label{i}
\bi = \call^{-1}(\theta) [\lambda - \psi_m \bc(\theta)]
\end{equation}
with the mappings
$$
\begin{aligned}
 \call(\theta) & ~:=~ L_s I_2 + {L_0 \over 2} Q(\theta)
 \\
 Q(\theta) & ~:=~ \lef[{cc} \cos(2 \theta) & \sin(2\theta) \\   \sin(2\theta) &  -\cos(2
\theta)\rig]
\\
\bc(\theta) & ~:=~ \col(\cos\theta, \sin\theta).
\end{aligned}
$$
Following  \cite{CHONAM} we define the active flux $\bx(t) \in \rea^2$ as
\begequ
\lab{actflu}
\bx:=\lambda - L_q \bi.
\endequ
As shown in \cite{CHONAM,ORTetal}, the rotor angle can be reconstructed from $\bx$ via the relation
\begequ
\label{tan_theta}
\tan(\theta)={\bx_2 \over \bx_1}.
\endequ

The central problem in sensorless control of electric motors is to estimate the mechanical angular position $\theta(t) \in {\mathbb S}$ from the measurable signals $\bv(t) \in \rea^2$ and $\bi(t) \in \rea^2$. This is equivalent to online estimation of the active flux $\bx$. \\

\noindent {\bf Observer Design Problem}
Consider the dynamical model \eqref{ipmsm1} with the output \eqref{i}. Design an observer
\begin{equation}
\begin{aligned}
    \dot \eta &~=~ F(\eta, \bi,\bv)
    \\
    \hat\bx &~=~ N(\eta,\bi,\bv)
\end{aligned}
\end{equation}
with $\eta(t)\in \rea^{n_\eta}$, that ensures global exponential convergence of the active flux estimation error, that is,
\begin{equation}
\label{conv:1}
\lim_{t\to\infty} |\hat\bx(t) - \bx(t)| = 0\quad\mbox{(exp.)}
\end{equation}
for all system and observer initial conditions. \hfill $\square$
%
%%%%%%%%%%%%%%%
\section{Assumptions and Preliminary Results}
\lab{sec3}
%%%%%%%%%%%%%5
%
The success of the GES observer design in \cite{ORTetal} relies on the following lemma, which was established in \cite{CHONAM}; see also \cite[Lemma 1]{ORTetal}. By introducing some LTI stable filters, we are able to obtain a perturbed linear regression equation (LRE) which is instrumental for the gradient observer design.

\begin{lemma}
\label{lem1}\rm
The electrical dynamics of the IPMSM \eqref{ipmsm1}, \eqref{i} satisfies the following (perturbed) linear regression equation
\begequ
\label{flux_dyn}
\begin{aligned}
y & = \Phi^\top \bx + d + \et,
\end{aligned}
\endequ
where the active flux $\bx$ is  defined in \eqref{actflu}, and the {\em measurable} signals $y(t) \in \rea^2$ and $\Phi(t) \in \rea^2$ are given as
\begin{equation}
\begin{aligned}
 y & ~:=~ L_0  H_2[\bi]^\top \Omega_1 + {1\over \alpha}|\Omega_1|^2 + {1\over \alpha} H_2[\Omega_2^\top \Omega_1]\\
\nonumber
 \Phi & ~:=~ \Omega_1 + \Omega_2.
 \lab{yphi}
 \end{aligned}
\end{equation}
\noindent with the signals $\Omega_1(t) \in \rea^2$ and $\Omega_2(t) \in \rea^2$ defined as
\begalis{
\Omega_1  &:= H_2 [\bv - R\bi - L_q p\bi]\\
 \Omega_2 &:= \Omega_1 - L_0 H_1[\bi],
} 
and the filters  
\begequ
\lab{h1h2}
H_1(p) := {\alpha p \over p+\alpha},\;H_2(p):= {\alpha \over p+\alpha},
\endequ 
where $\alpha>0$ is a {\em tuning} parameter.

The (unknown) perturbing signal $d$ is given by
\begequ
\lab{d}
 d := - \ell H_1 \Bigg[\bi^\top {\bx \over |\bx|}\Bigg],
\endequ
with  $\ell:=\psi_m L_0$, and $\et\in \rea^2$ is an exponentially decaying term\footnote{Following standard practice, we neglect these terms in the sequel.}  caused by the initial conditions of the filters.\hfill $\square$
\end{lemma}

To address the observer design problem, we make the following assumptions that, as thoroughly discussed in \cite{CHONAM,ORTetal}, hold true in practice---see Remark \ref{rem2} below. 

\begin{assumption}\rm 
\label{ass:1} ({\em Motor rotating behavior})
$\Phi$ is persistently excited, namely, there exist $T>0$ and $\delta>0$ such that
\begequ
\label{PE}
\int_t^{t+T} \Phi(s) \Phi^\top(s) ds \ge \delta I_2, \quad \forall t \ge0
\endequ
with $|\Phi|$ upper bounded.
\end{assumption}
\begin{assumption}\rm 
\label{ass:2} ({\em Boundedness})
The motor operates in a mode guaranteeing that all signals $\bi$, $\bv$ and $\lambda$ are bounded and $|\bx| \ge x_{\tt min}$ for some constant $x_{\tt min} >0$.
\end{assumption}
\begin{assumption}\rm
\lab{ass:3} ({\em Small anisotropy})
The current $\bi$ verifies $|L_0 \bi | < \psi_m$.
\end{assumption}

Motivated by the flux dynamics \eqref{ipmsm1} an observer of the following form is suggested in many works
\begin{equation}
\label{observer}
\begin{aligned}
	\dot{\hat \lambda} & = \bv - R\bi + E\\
	\hat \bx & = \hat \lambda - L_q \bi \\
\end{aligned}
\end{equation}
with a signal $E(t) \in \rea^2$ to be defined. Invoking \eqref{actflu} we see that, for this observer structure, we get 
\begequ
\lab{dottilx}
\dot{\tilde \bx} = E
\endequ
with $\tilde \bx:= \hat\bx - \bx$. Several different correction terms $E$ can be used to complete the observer design.\\

\noindent {\bf i)} By leaving out the perturbation term $d$ and applying a gradient descent to minimize the cost function 
	$$
	J_{\tt tie}(\hat\bx) = \Big|y - \Phi^\top \hat\bx \Big|^2,
	$$
	the correction term in \cite{CHOetal}\footnote{The paper \cite{CHOetal} considers another perturbed linear regression model rather than \eqref{flux_dyn}, but they are exactly in the same mathematical form. Hence, we may conclude the same stability properties for the closed-loop dynamics.} is selected as
	$$
	E_{\tt tie} = \gamma \Phi (y - \Phi^\top \hat \bx), \quad \gamma>0.
	$$
As shown in  \cite{CHOetal} the estimation error dynamics is {\em practically} exponentially stable with bounded ultimate errors.\\
	
\noindent {\bf ii)}  The paper \cite{ORTetal} proposes the correction term\footnote{In \cite{ORTetal} a slight modification in $\hat d$ is introduced to guarantee $|\hat\bx|>x_{\tt min}$ and thus avoid singularities.}
\begin{equation}
\label{E1}
\begin{aligned}
E_{\tt aut} & ~=~  \gamma \Phi \Big( y -\Phi^\top \hat\bx - \hat d \Big), \quad \gamma>0
\\
\hat d &~=~ - \ell H_1 \Bigg[\bi^\top {\hat\bx \over |\hat\bx|}\Bigg].
\end{aligned}
\end{equation} 
 The term in \eqref{E1} can be viewed as a certainty equivalent approximation of the gradient descent to minimize the cost function
\begin{equation}
\label{J2}
J_{\tt aut}(\hat\bx) = \Big|y- \Phi^\top \hat\bx - \Hat{d} \Big|^2.
\end{equation}

In  \cite{ORTetal}, to guarantee global exponential stability, we require the adaptation gain $\gamma>0$ to be {\em sufficiently small}, which may have a deleterious effect on the transient performance. The main objective of this paper is to propose an alternative observer design that {\em removes} the requirement of small adaptation gain.

\begin{remark}\rm 
In Lemma \ref{lem1}, we obtain a nonlinear regression model \eqref{flux_dyn} via the filtering technique. It can be viewed as a linear regression model with respect to $\bx$ with a perturbation term $d$. This term $d = - \ell H_1 [\bi^\top {\bx \over |\bx|}]$ is, in nature, the filtered signal of the inner product between the current (stator flux) and the \emph{normalized} rotor flux vector, which should be maintained to be zero for torque maximization. %However, it does not provide any information of the reluctance torque.
\end{remark}
%%%%%%

\begin{remark}\rm 
\lab{rem2}
Some remarks on the assumptions adopted in the paper are in order.

\noindent {\bf (i)} The persistency of excitation condition \eqref{PE} is generally satisfied when motors are operating at middle- or high-speed regions. Otherwise, it is necessary to probe high-frequency signals in low speed to impose observability \cite[Section 5]{ORTetal}. Note that another way to use signal injection is to demodulate the angular information implicitly carried by the high-frequency component of stator currents \cite{NAM,YIetal}. 

\noindent {\bf (ii)} Assumption \ref{ass:2} is reasonable since the active flux $\bx$ is not zero as far as the rotor permanent magnet is magnetized.
\end{remark}
%
%%%%%%%%
\section{Main Result}
\label{sec3}
%%%%%%%%
%
To overcome the performance limitation imposed by the requirement to use a small adaptation gain in the gradient-based scheme of \cite{ORTetal} we add to the filtering technique of \cite{CHONAM,ORTetal}, a novel mathematical concept: the creation of a {\em virtual invariant manifold}. Thanks to this new technique we are able to replace the gradient-based observer design by the high performance KRE-based estimator. As explained in \cite{ORTetalARC} the key step in Kreisselmeier's estimator is the construction of the KRE that, as shown in \cite{ORTetalTAC}, is a particular form of the more general dynamic regressor extension. For the purpose of this paper, the key feature of this estimator is that, in contrast with gradient or least-squares schemes, it is possible to make the estimator converge arbitrarily fast (after a short transient phase due to filtering) by increasing the adaptation gain---see also \cite{ARAetal} for a detailed analysis of the properties of this estimator.  

Motivated by Lemma \ref{lem1}, Kreisselmeier's estimator \cite{KRE}, the results of \cite{Li}, our previous work \cite{ORTetal} and the discussion in the previous section, we propose the following flux observer.

\emph{Flux and position observer}
\begin{equation}
\label{KRE}
\begin{aligned}
&\left.
    \begin{aligned}
	\dot Q & = - a\big(Q- \Phi \Phi^\top \big),\;Q(0)=0\\
\dot Y 	& =  - a\big(Y - \Phi e\big) +QE,\;Y(0)=0	\\
E&=-\gamma Y \\
\end{aligned}
~~ \right\} ~\mbox{(KRE)}
\\
&\left.
\begin{aligned}
	\dot{\hat \lambda} & ~= \bv - R\bi + E\\
	\hat \bx & ~= \hat \lambda - L_q \bi \\
 \hat \theta & ~=  { {\rm atan2} (\hat{\bx}_2, \hat{\bx}_1), }
\end{aligned}
\hspace{.9cm} \right\}
~~\mbox{(Flux-position estimate)}
\end{aligned}
\end{equation}
with the estimate of the disturbance signal
\begin{equation}
\label{e}
\begin{aligned}
\hat d &~=~ - \ell H_1 [\bi^\top \sigma(\hat\bx)],
\end{aligned}
\end{equation}
the variable
$$
 e:= \Phi^\top \hat \bx + \hat d - y,
$$
and the mapping
$$
\sigma(\hat\bx) = \left\{
\begin{aligned}
& ~~{\hat\bx \over |\hat\bx|}  & \qquad \text{if~~} |\hat\bx| \ge \epsilon>0 \\
& ~~\col(0,0) & \qquad \text{otherwise,}
\end{aligned}
\right.
$$
where $\epsilon\in(0,x_{\tt min})$, and $a, \alpha ,\gamma>0$ are tuning parameters. \hfill $\square$

We summarize the provable properties of the above observer as follows.

\begin{proposition}\rm\label{prop:1}
Consider the model \eqref{ipmsm1} with output \eqref{i}. For any adaptation gains $\gamma>0$ and $a>0$, there always exists a scalar $\alpha_{\max}>0$ such that the observer \eqref{KRE} and \eqref{e} guarantees the global exponential convergence \eqref{conv:1} and
$$
\lim_{t\to\infty} 
\left|\hat\theta(t) - \theta(t)  \right|
= 0 \quad \mbox{(exp.)}
$$
{\em for any}\footnote{We underscore the fact that, as shown in \eqref{h1h2}, $\alpha$ is a free tuning parameter that determines the time constant of the low pass LTI filters $H_1(p)$ and $H_2(p)$.} $\alpha \leq \alpha_{\max}$.
\end{proposition}
%%%

\begin{proof}\rm 
 First, it is shown in \cite[Section III-A]{ARAetal} that, given {\bf Assumption 1}, the square matrix $Q$ is positive definite for all $t\ge T$, i.e.,
 \begequ
 \lab{bouq}
 Q(t) \ge q I_2, \quad \forall t \ge T,
 \endequ
for some $q>0$. 
 
The next key step is to show the existence of a (virtual) forward invariant manifold
$$
\calm:= \{(Y,Q,\xi) \in \rea^2 \times \rea^{2\times 2} \times {\rea^2}: Y = Q\tilde \bx + \xi\},
$$
in which the variable $\xi$ is governed by the dynamics
\begin{equation}
\label{id:1}
\dot \xi = - a\big(\xi - \Phi \tilde d \big),\;\xi(0)=0,
\end{equation}
and we have defined $\tilde d : =\hat d -d$. To verify this we define the error signal
$$
\pi:= Y -(Q\tilde \bx + \xi),
$$
whose derivative is given as
and compute its time derivative as
\begalis{
	\dot \pi & ~=~\dot Y - \dot Q \tilde \bx - Q\dot{\tilde{\bx}}- \dot\xi \\
	&~=~ - a\big(Y - \Phi e\big) +QE + a \big( Q -\Phi\Phi^\top \big)\tilde \bx - QE\\
        & ~~~\quad  + a\big(\xi - \Phi \tilde d\big)\\
	& ~= ~- a(Y- Q\tilde \bx -\xi) \\
	& ~=~ - a \pi.
}
Consequently, since for all $\tilde x(0)$ we have that
$$
\pi(0)=Y(0)-Q(0)\tilde x(0)-\xi(0)=0,
$$
we complete the proof of the invariance claim and show that
\begequ
\lab{manbeh}
Y(t) = Q(t)\tilde x(t) + \xi(t),\;\forall t \geq 0.
\endequ

The last step in the proof is to analyze the closed-loop dynamics. The active flux estimation error $\tilde \bx$ satisfies
\begalis{
\dot {\tilde \bx}&=E\\
&=-\gamma Y\\
&= -\gamma Q\tilde \bx -\gamma \xi,
}
where we have used \eqref{manbeh} to get the last identity. Now, the variable $\tilde d$ is given by
\begalis{
\tilde d 
& = -\ell H_1 \left[\bi^\top \left({\hat \bx \over |\hat \bx|}-{\bx \over |\bx|}\right)\right],
}
which, as shown in \cite{ORTetal}, may be written as
\begali{
\tilde d 
& = - H_1[\bw^\top \tilde\bx] \lab{tild}
}
with the signal $\bw(t) \in \rea^2$ defined via
$$
\bw
~=~\left\{
\begin{aligned}
& {\ell \over |\hat \bx|} \Bigg[ I_{2}
 -    { (\bx + \hat\bx) \bx^\top \over |\bx|(|\bx|+|\hat\bx|)} \Bigg] \bi,  &&
 \mbox{if~} |\hat\bx|\ge \epsilon
 \\
& - \ell{\bi^\top \bx \over |\bx| |\tilde\bx|^2} \tilde\bx , && \mbox{otherwise.}
\end{aligned}
\right.
$$
A state space realization of \eqref{tild} is given by
\begalis{
\dot z&=-\alpha (z-\alpha \bw^\top\tilde\bx) \\
\tilde d & = z-\alpha \bw^\top\tilde\bx.
}
Defining the state vector $\chi(t) := \col(\tilde \bx(t), \xi(t), z(t)) \in \rea^5$ we obtain the closed-loop dynamics 
\begin{equation}
\label{dot:chi}
\dot \chi
=
\begmat{
- \gamma Q & - \gamma I_2 & 0_{2 \times 1} \\
-a\alpha \Phi \bw^\top & - aI_2 &  a\Phi\\
\alpha^2 \bw^\top & 0_{1 \times 2} & -\alpha
}\chi.
\end{equation}

The dynamics \eqref{dot:chi} may be written as a perturbed linear time-varying (LTV) system of the form
$$
\dot \chi=A(t)\chi+\alpha \Delta(t)\chi,
$$
where we defined the matrices
$$
\begin{aligned}
A(t) & ~:=~ \begmat{
- \gamma Q(t) & - \gamma I_2 & 0_{2 \times 1} \\
0_{2 \times 2} & - aI_2 &  a\Phi(t)\\
0_{1 \times 2} & 0_{1 \times 2} & -\alpha
}\in \rea^{5 \times 5}
\\
\Delta(t) & ~:=~
\begmat{
0_{2 \times 2} & 0_{2 \times 2} & 0_{2 \times 1} \\
-a \Phi(t) \bw^\top(t) & 0_{2 \times 2} &  0_{2 \times 1}\\
\alpha \bw^\top(t) & 0_{1 \times 2} & 0
}\in \rea^{5 \times 5}
\end{aligned}
$$
We make the observation that $\Phi(t)$ and $\bw(t)$ are both bounded. Hence, from \eqref{KRE} we have that $Q(t)$ is also bounded. As a result, the above LTV system is forward complete during $[0,T)$. Moreover, from Assumptions \ref{ass:1} and \ref{ass:2}, we have that $Q(t)$ is positive definite for all $t\ge T$, that is, it satisfies \eqref{bouq}. We proceed with our analysis for $t\ge T$.

Let us introduce the following parameters
$$
\mu_\Phi:=\|\Phi\|_\infty , \; \mu_w:=\|\bw\|_\infty.
$$
Consider the candidate Lyapunov function 
$$
V(\chi):=\hal \chi^\top P \chi,
$$
where we defined the positive definite matrix
$$
P:=\diag\left\{{1\over \gamma},{1\over \gamma},{1\over aq},{1\over aq}, {\mu_\Phi^2 \over q\alpha}\right\}.
$$
Its time derivative satisfies
$$
\begin{aligned}
\dot V & ~= ~
- \|\tilde \bx\|^2_Q -  \tilde \bx^\top \xi - {1\over q} \left( |\xi|^2  -  \xi^\top \Phi z +  \mu_\Phi^2 |z|^2 \right)
\\
& \qquad
+\alpha \chi^\top P\Delta \chi
\\
& ~\le ~
- q |\tilde \bx|^2 -  \tilde \bx^\top \xi - {1\over 2q} \left( |\xi|^2   -  \mu_\Phi^2 |z|^2 \right)+\alpha \rho | \chi|^2
\\
& ~\le ~
- q |\tilde \bx|^2 + {3q\over 4}|\tilde \bx|^2 + {1\over 3q}|\xi|^2 - {1\over 2q} \left( |\xi|^2   -  \mu_\Phi^2 |z|^2 \right)
\\
& \qquad
+\alpha \rho | \chi|^2
\\
& ~\le ~
 -  {q\over 4}|\tilde \bx|^2  - {1\over 6q} |\xi|^2 - {\mu_\Phi^2\over 2q} |z|^2 +\alpha \rho | \chi|^2 \\
& ~\le ~
 -  (\beta -\alpha \rho) | \chi|^2
\end{aligned}
$$
where
$$
\beta:=\min\Big\{{q\over 4}, {1\over 6q},{\mu_\Phi^2\over 2q}\Big\}
$$
and we used the fact that
$$
\chi^\top P\Delta \chi={1\over q}(\xi^\top \Phi+\mu_\Phi^2 z)\bw^\top \tilde x \leq \rho | \chi|^2,
$$ 
for some $\rho>0$  {\em independent} of $\alpha$. Consequently, there exists a sufficiently small $\alpha$ such that   $\beta -\alpha \rho>0$, and the proof is concluded invoking the last upper bound on $\dot V$ above, as well as the algebraic equation \eqref{tan_theta}. 
\qed
\end{proof}

\begin{remark}
\rm
There are three positive adaptation parameters to tune, i.e., $a,\alpha$ and $\gamma$. 

\noindent {\bf (i)} The parameter $\alpha$ appears in the filters $H_1(p)$ and $H_2(p)$ for pre-filtering in order to generate the regression model \eqref{flux_dyn}. This parameter affects the convergence rate of the exponentially decaying term $\et$. As shown in Proposition \ref{prop:1}, a large $\alpha>0$ may yield instability of the closed loop. 
	
\noindent {\bf (ii)} 	The parameter $a$ appears in the KRE $\dot Q  = -a(Q-\Phi\Phi^\top)$. Roughly, a small $a>0$ means that ``more past'' information of $\Phi\Phi^\top$ is utilized in the matrix $Q$.

\noindent {\bf (iii)} The parameter $\gamma$ appears in the $(1,1)$-element of $A(t)$, making it closely connected to the convergence rate of the active flux estimate. In the above design we fix $Q(0)=0$ and $Y(0)=0$ and $\xi(0)=0$. If we select a small $a>0$ and a sufficiently large $\gamma$, then the estimation error $\tilde \bx$ will converge to a small neighborhood of zero very fast. In contrast, in \cite{ORTetal} we require the adaption gain $\gamma>0$ \emph{sufficiently small}, thus we are unable to assign the transient performance.
\end{remark}
%%%%
\begin{remark}
\rm 
{The correction term $E$ in the proposed flux and position observer \eqref{KRE} can be roughly viewed as the gradient descent of the cost function generated from the KRE, rather than the ``natural'' cost function \eqref{J2} adopted in \cite{ORTetal}. This idea was originally proposed in the pioneer work \cite{KRE} to improve transient performance in adaptive observers, for which the unknown parameters are \emph{constant}. In the context of observer design, we need to take into account the dynamics of the \emph{time-varying} states, and this is captured by the additional variable $\xi$ in the proof.  }
\end{remark}

%\rom{I stopped here because, as far as I understood, there are still some change to be made to the section.} 

%%%%%%%%%%

\section{Simulations and Experiments}

In this section, we present some simulation and experimental results to valid the performance of the proposed flux and position observer.

\subsection{Simulation Results}
\label{sec:51}

In simulations, in order to make it more realistic we adopted the same motor parameters as those used in experiments. They are $R=2.5 ~\Omega$, $L_d = 0.00782$ H, $L_q = 0.00782$ H, and $\psi_m = 0.10$, with the pole number equal to 8. We consider the gains $\alpha = 200\pi$ and $ a = 20\pi$, and consider different adaptation gains $\gamma$. Simulations were performed using MATLAB/Simulink. The motor was rotating at 1000 rpm, and we selected the initial magnetic flux angle error as $\pi/2$. The magnitude error was set as twice as the real value, i.e., $\hat{\lambda}(0) = 2\psi_m [\cos(\theta - \pi/2), \sin(\theta - \pi/2)]^\top$.

In Fig. \ref{fig:3}, we present the simulation results for the observer \eqref{KRE} with the adaptation gain $\gamma =1$, showing its satisfactory performance. By increasing $\gamma$ to 5, we observe in Fig. \ref{fig:1} that the convergence rate becomes faster as expected from Proposition \ref{prop:1}. This is compared to the position observer in our previous work \cite{ORTetal}, with the simulation results in the same scenario shown in Figs. \ref{fig:4} and \ref{fig:2} for the gains $\gamma =1$ and $\gamma =5$, respectively. Note that the approach in \cite{ORTetal} requires that the adaptation gain $\gamma$ sufficiently small to achieve global exponential stability. It is clear in Fig. \ref{fig:2} that when selecting $\gamma =5$ it slows down the convergence rate and even causes the inconsistency issue. This demonstrates the key merit of tunability for the observer proposed in this paper.

\begin{figure}[!htp]
   \centering
   \includegraphics[width=0.4\textwidth]{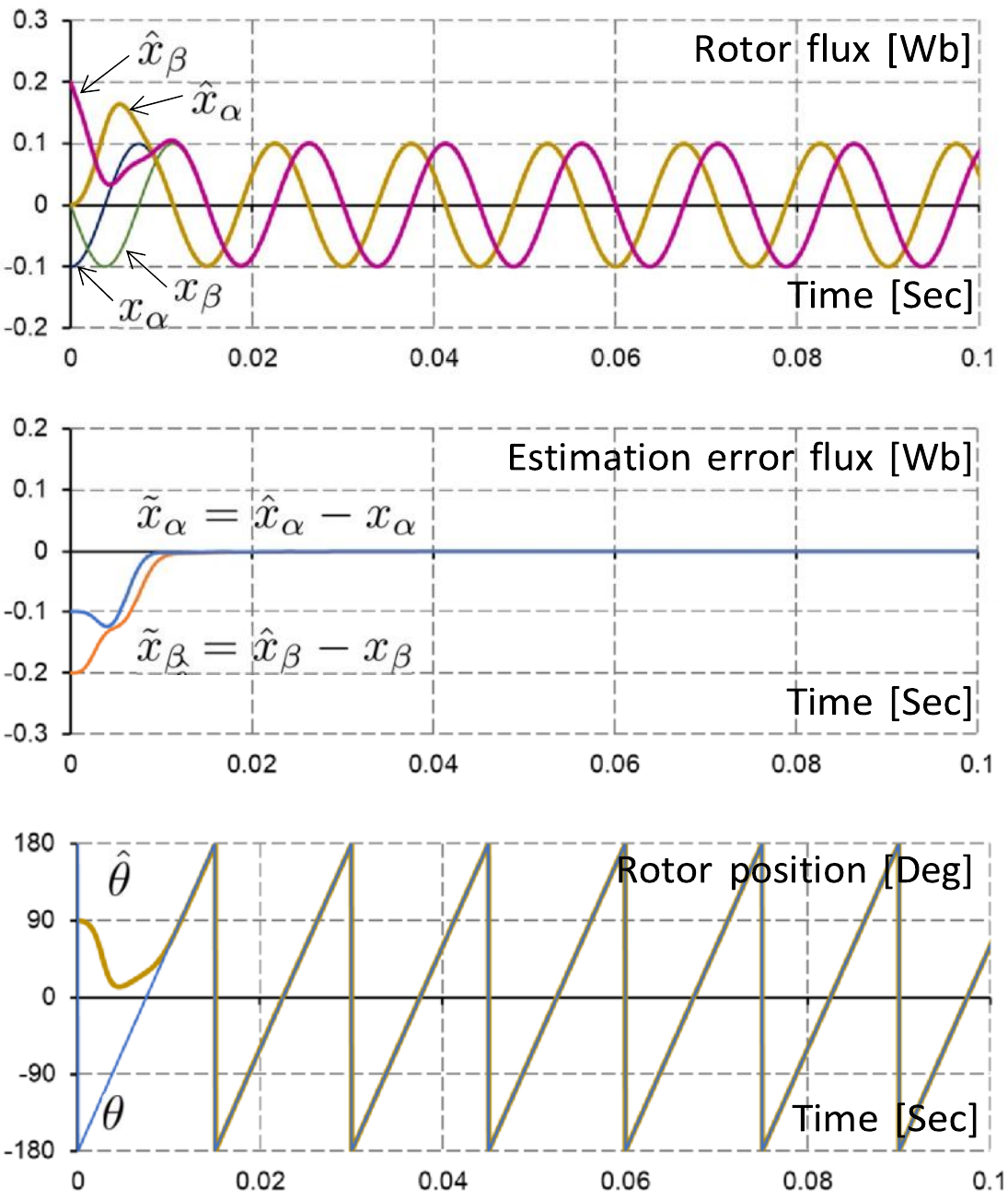}
   \caption{Simulation results: Estimation of rotor flux and angle from the proposed observer with $\gamma = 1$}
    \label{fig:3}
\end{figure}
%%%
\begin{figure}[!htp]
   \centering
   \includegraphics[width=0.4\textwidth]{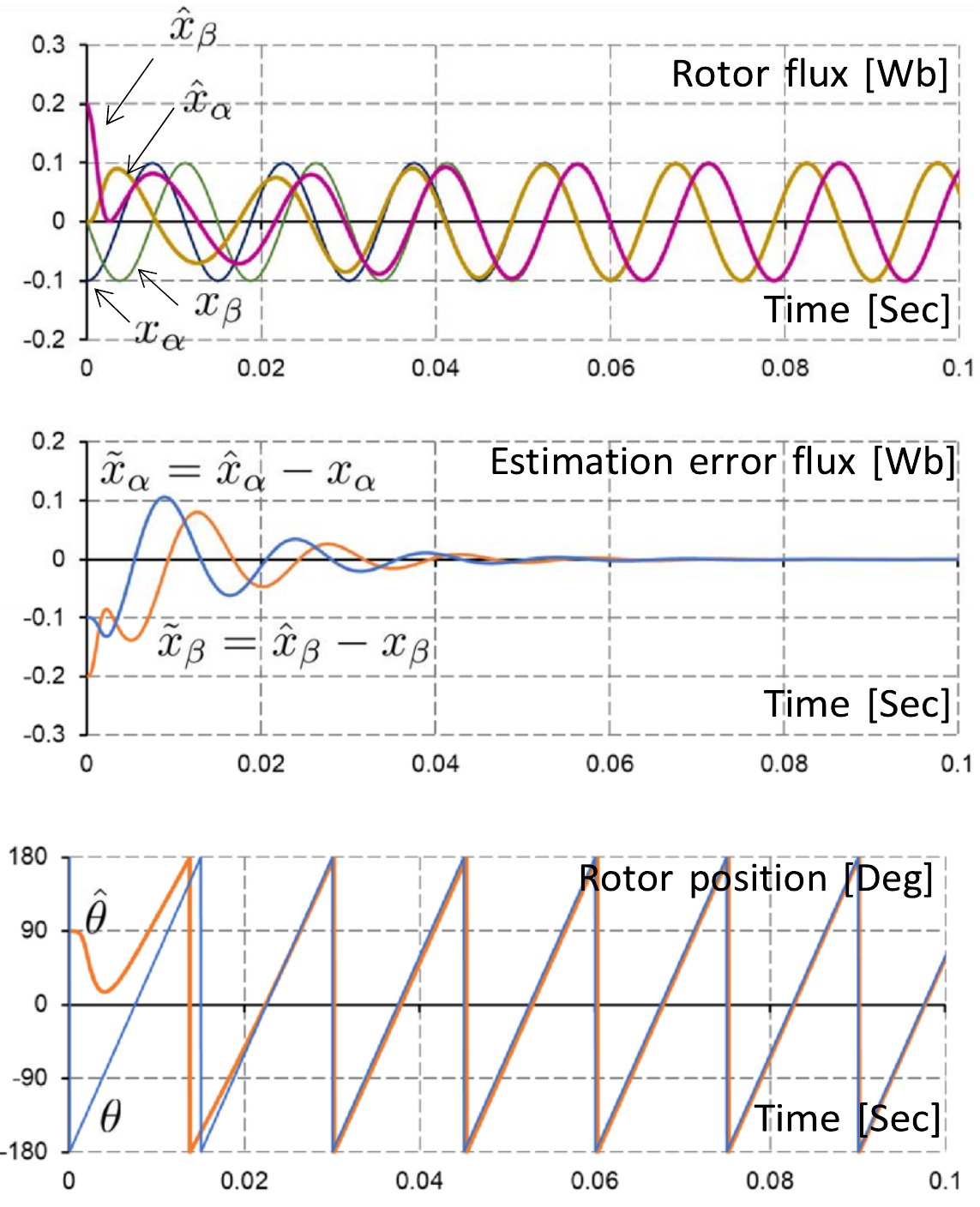}
   \caption{Simulation results: Estimation of rotor flux and angle from the observer in \cite{ORTetal} with $\gamma = 1$}
    \label{fig:4}
\end{figure}

\begin{figure}[!htp]
   \centering
   \includegraphics[width=0.4\textwidth]{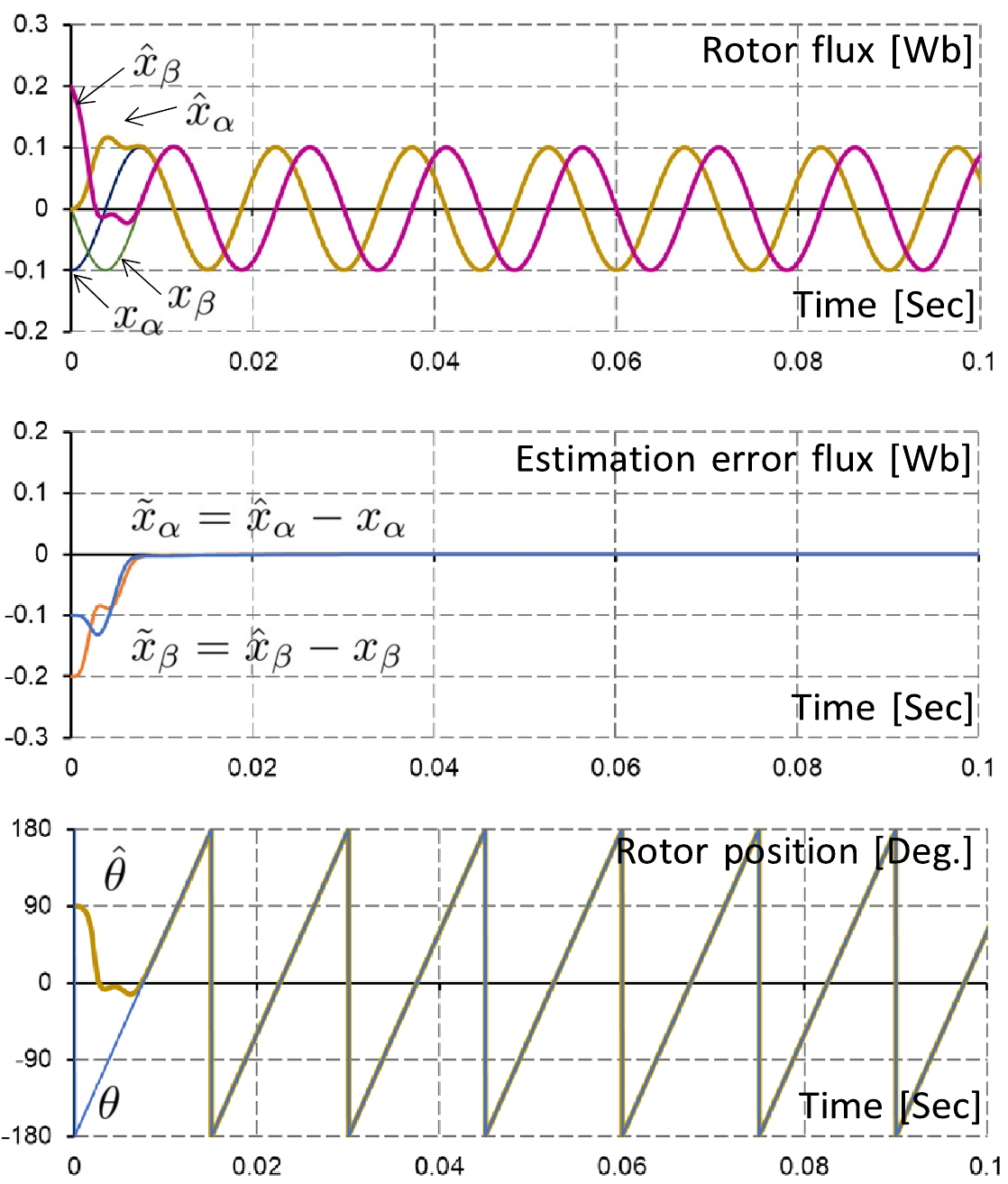}
   \caption{Simulation results: Estimation of rotor flux and angle from the proposed observer with $\gamma = 5$}
    \label{fig:1}
\end{figure}
%%%
\begin{figure}[!htp]
   \centering
   \includegraphics[width=0.4\textwidth]{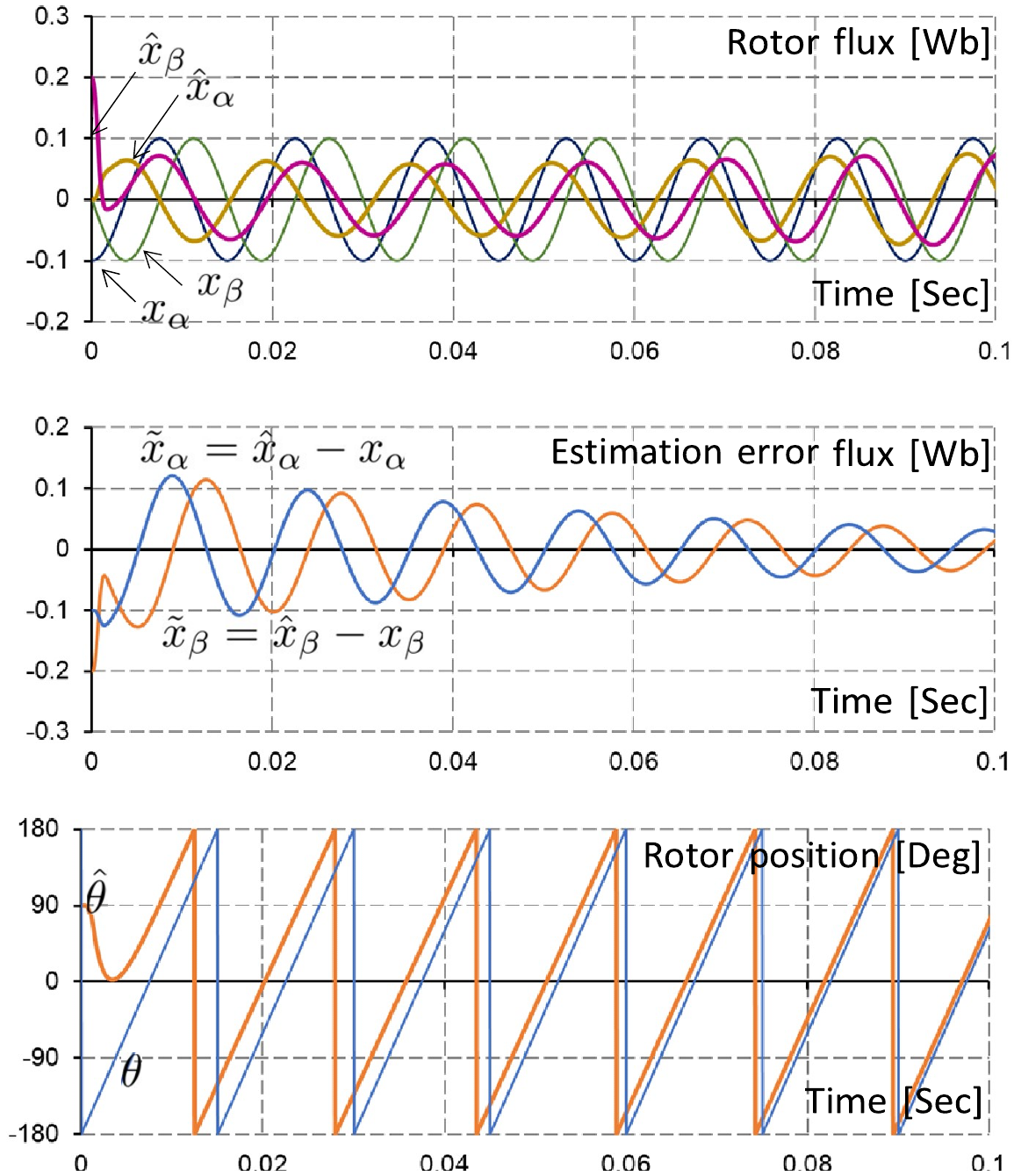}
   \caption{Simulation results: Estimation of rotor flux and angle from the observer in \cite{ORTetal} with $\gamma = 5$}
    \label{fig:2}
\end{figure}

\begin{figure}[!htp]
   \centering
   %%%%%%%%%%%%%%%%%%
   \subfigure[SiC MOSFET Based Inverter]{
   \includegraphics[width=0.35\textwidth]{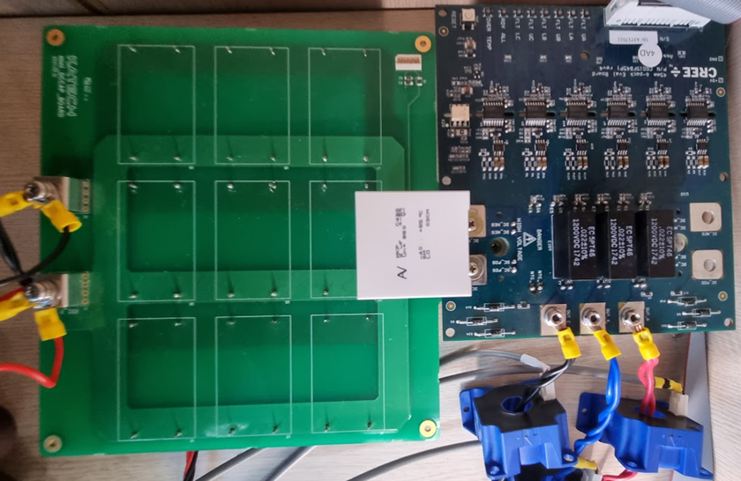}
   \label{fig:s1}
   }
   %%%%%%%%%%%%%%%%%%
   \subfigure[Dynamo set]{
   \includegraphics[width=0.35\textwidth]{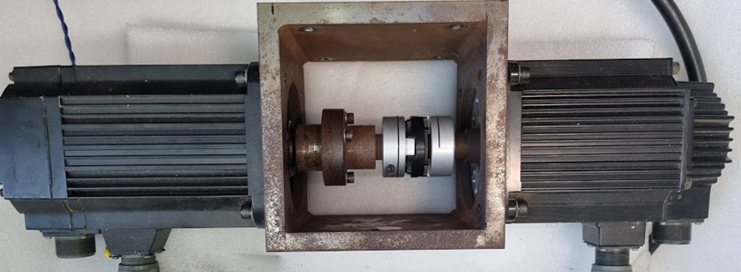}
   \label{fig:s2}
   }
   %%%%%%%%%%%%%%%%%%
   \caption{Experimental environment}
    \label{fig:setup}
\end{figure}

%%%%%%%%%%%%%%%%%%%%%%%%%%

\begin{figure*}[!htp]
   \centering
   %%%%%%%%%%%%%%%%%%
   \subfigure[Estimation of rotor flux with $\gamma = 1$]{
   \includegraphics[width=0.315\textwidth,height=3cm,frame]{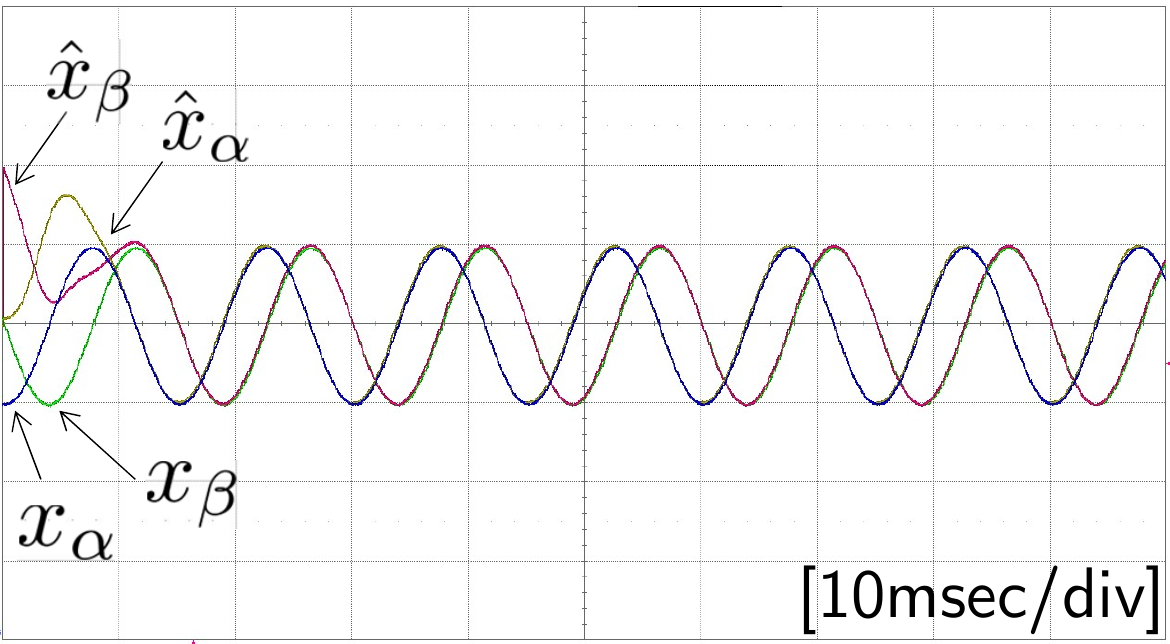}
   \label{fig:5a}
   }
   %%%%%%%%%%%%%%%%%%
   \subfigure[Estimation of rotor flux with $\gamma = 5$]{
   \includegraphics[width=0.315\textwidth,height=3cm,frame]{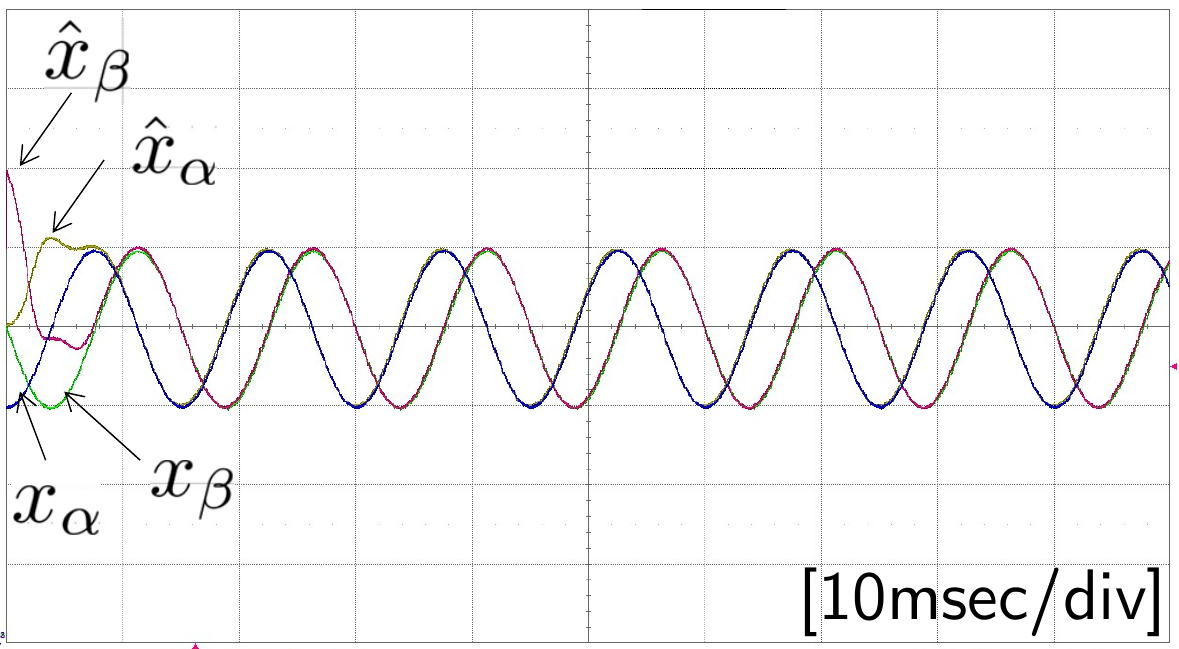}
   \label{fig:5b}
   }
   %%%%%%%%%%%%%%%%%%
   \subfigure[Estimation of rotor angle with $\gamma = 5$]{
   \includegraphics[width=0.315\textwidth,height=3cm,frame]{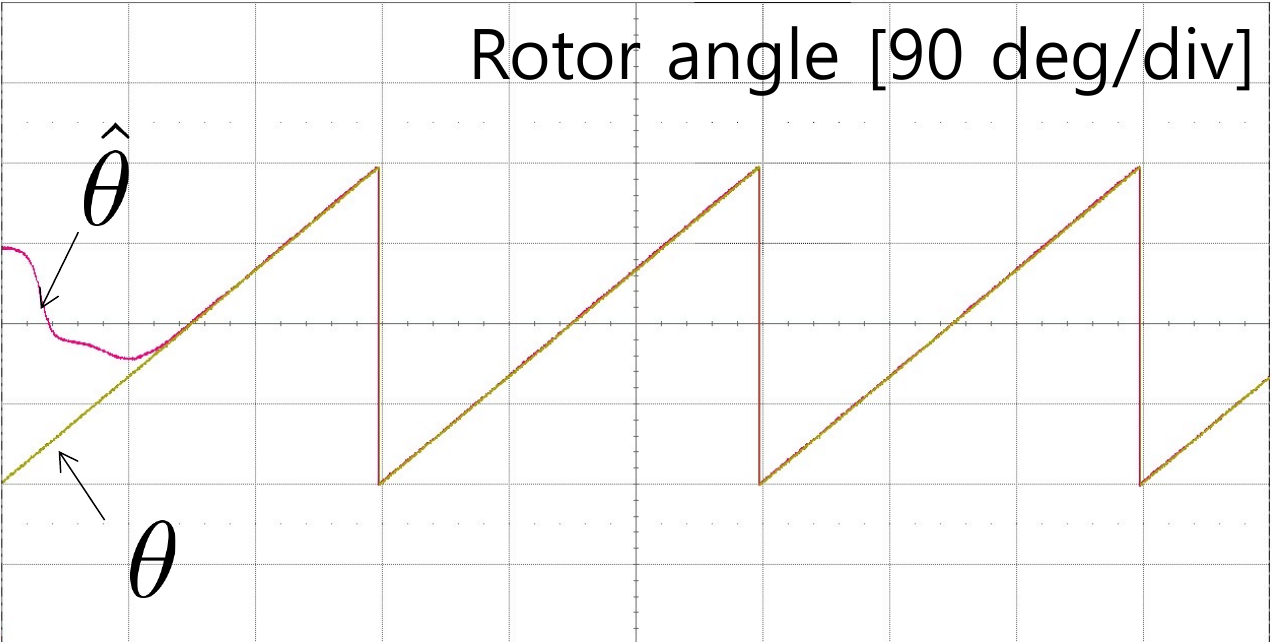}
   \label{fig:5c}
   }
   %%%%%%%%%%%%%%%%%%
   \caption{Experimental results: Performance of the proposed observer}
    \label{fig:5}
\end{figure*}
%
%%%%%%%%%%%%%%%%%%%%%%%%%%
%
\begin{figure*}[!htp]
   \centering
   %%%%%%%%%%%%%%%%%%
   \subfigure[Estimation of rotor flux with $\gamma = 1$]{
   \includegraphics[width=0.315\textwidth,height=3cm,frame]{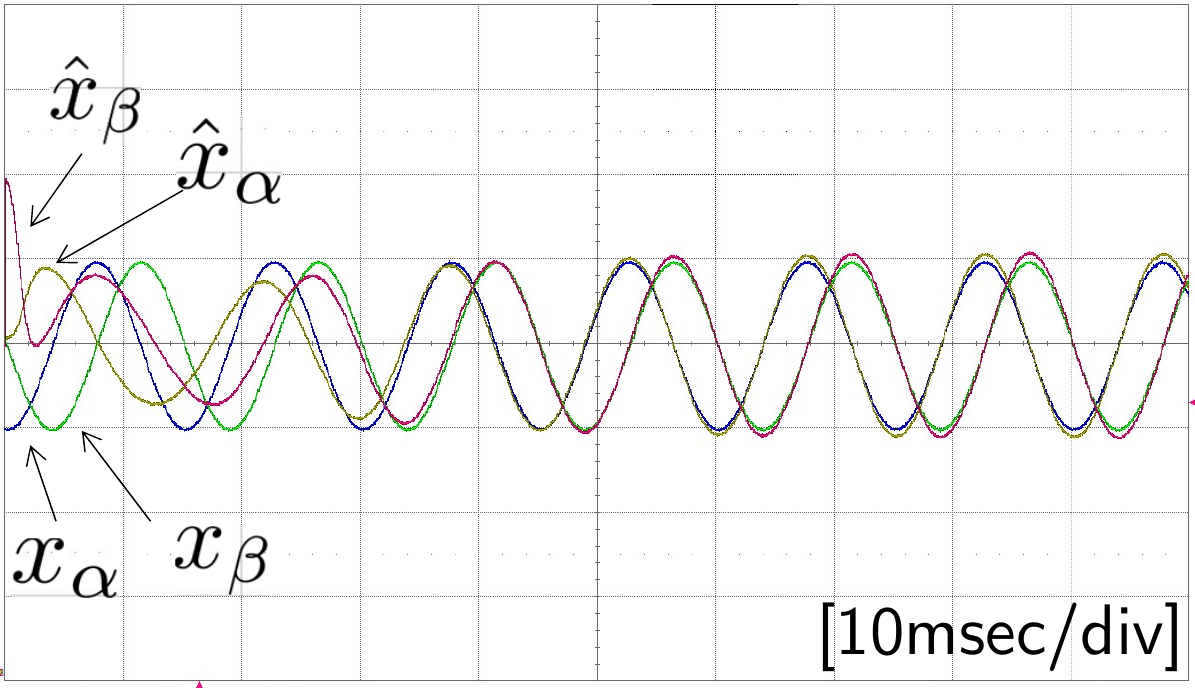}
   \label{fig:6a}
   }
   %%%%%%%%%%%%%%%%%%
   \subfigure[Estimation of rotor flux with $\gamma = 5$]{
\includegraphics[width=0.315\textwidth,height=3cm,frame]{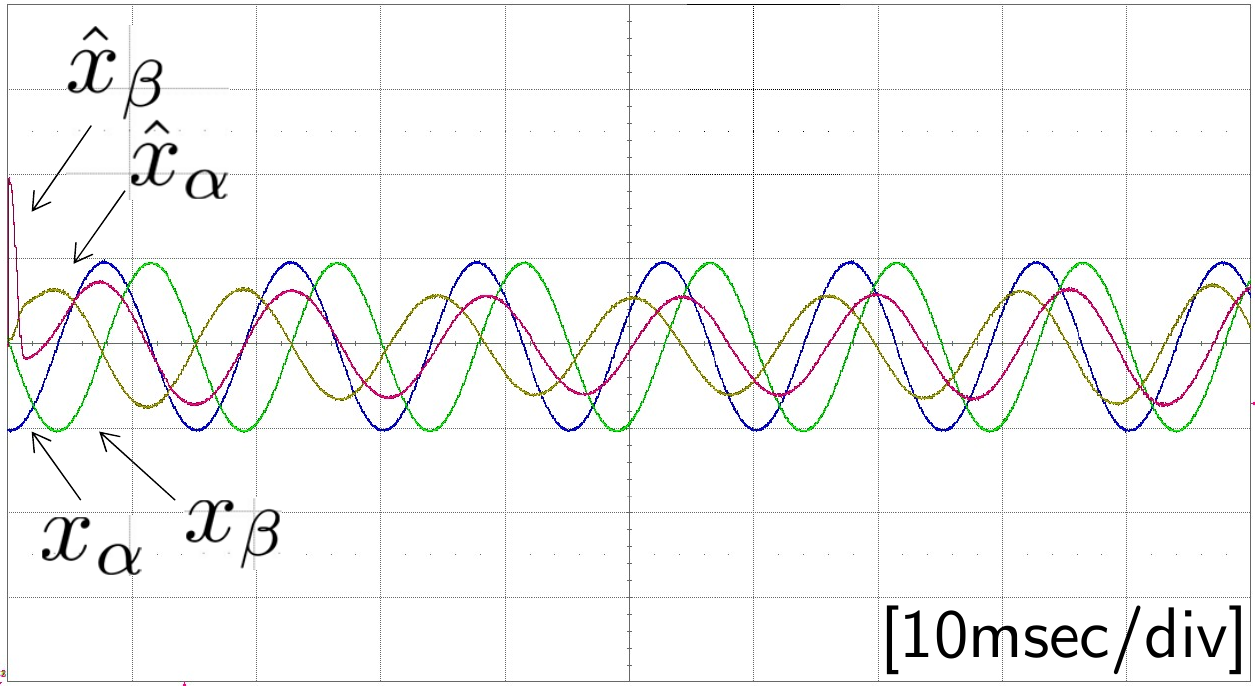}
   \label{fig:6b}
   }
   %%%%%%%%%%%%%%%%%%
   \subfigure[Estimation of rotor angle with $\gamma = 5$]{
\includegraphics[width=0.315\textwidth,height=3cm,frame]{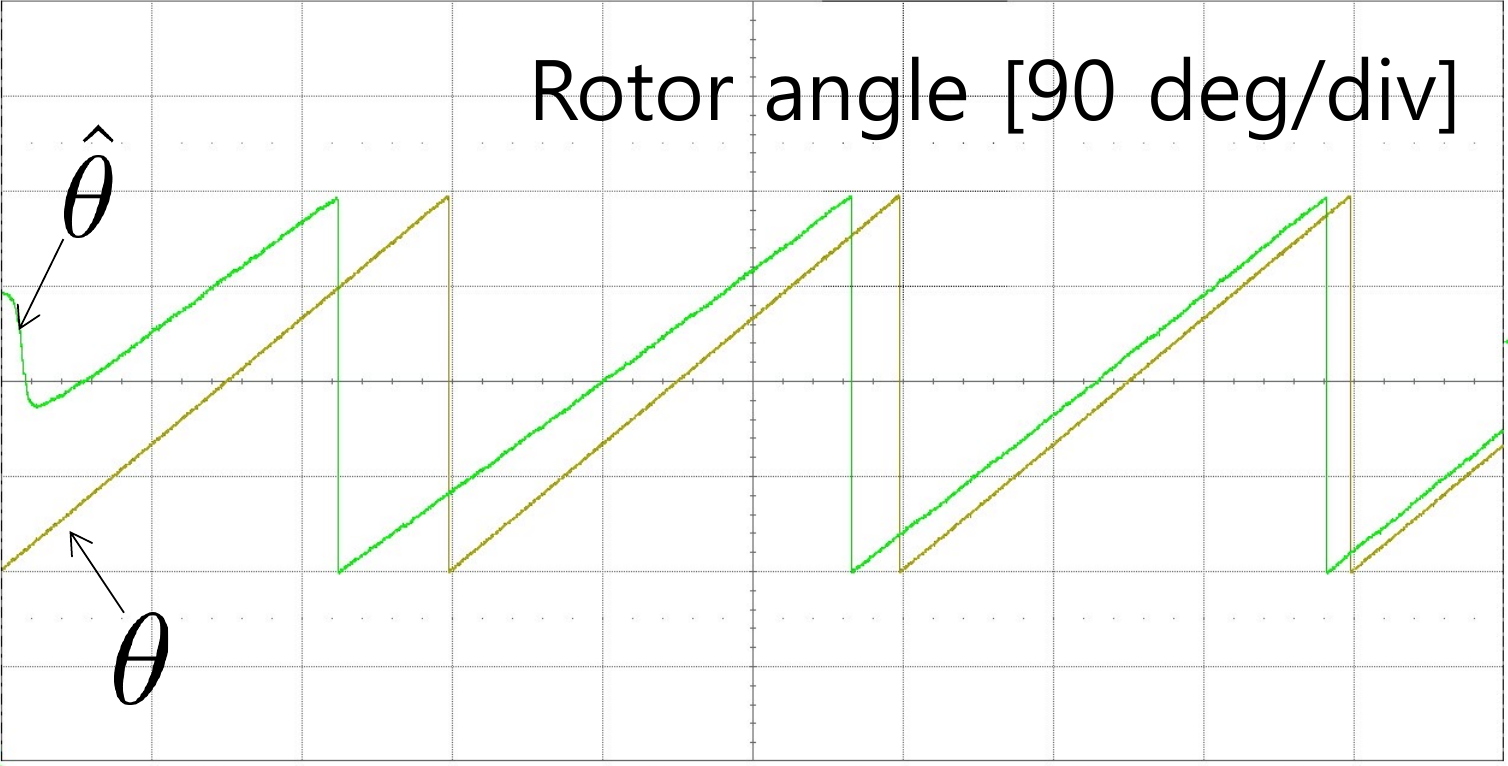}
   \label{fig:6c}
   }
   %%%%%%%%%%%%%%%%%%
   \caption{Experimental results: Performance of the observer in \cite{ORTetal}}
    \label{fig:6}
\end{figure*}

\subsection{Experiments}

Some experiments have been conducted on the test bench in Fig. \ref{fig:setup}, and the motor parameters have been given in Section \ref{sec:51}. The experimental scenario is that the test motor rotated at 1000 rev/min. The motor started initially without any control. Then, the motor drive started, and the shaft of the motor rotates at 1000 rpm. The rotor flux observer was operated from incorrect initial values. We selected the parameters $\alpha = 20\pi$ and $\alpha = 200\pi$. The inverter was made with Silicon Carbide Six-Pack Power Modules (CCS050M12CM2) and gate driver(CGD15FB45P1). The inverter is controlled by DSP TMS320F28337. The switching frequency was set 5 kHz and the current sampling frequency was 10 kHz.

In Fig. \ref{fig:5} we show the experimental results of the proposed flux and position observer in this paper, illustrating its capability to reconstruct the flux and the angular position as well with high accuracy. Similarly to the simulation part, we tested the observer with different gains $\gamma$, and the result shows the better transient performance when using a high adaption gain $\gamma$. Whereas, the steady-state error performance of the proposed observer does not have a significant difference with respect to the parameter $\gamma$. We also used the same scenario to test the observer in \cite{ORTetal} with experimental results shown in Fig. \ref{fig:6}. From Fig. \ref{fig:6a}, we observe that its convergence time is longer than the one for the new design, and by tuning a larger $\gamma =5$ there are significant ultimate errors of flux and angular estimates with degraded transient behaviour; see Figs. \ref{fig:6b}-\ref{fig:6c}. In this way, we have verified the superior of the proposed observer, which enjoys enhanced transient performance and provable stability.

%%%%%%%%
\section{Conclusion}
%%%%%%%%

In this paper, we revisited the problem of design a GES observer for sensorless control of IPMSMs. Compared to our previous result in \cite{ORTetal}, the new design does not impose a requirement of using a sufficiently small adaptation gain $\gamma>0$, a restriction that may affect the transient performance of the observer. Our motivation in this paper is to remove this restriction to improve the convergence rate of the observer. 

The success of the new design relies on the introduction in the proof of a virtual invariant manifold that is instrumental for the inclusion of the KRE estimator.  It is interesting to note that the proof of our new design is significantly simpler than the one given in \cite{ORTetal}. We have also conducted comprehensive simulations and real-world experiments to validate our theoretical development and its practical usefulness.

%%%%%%%%


\begin{thebibliography}{aa}
%%%%%%%%

\bibitem{ARAetal}
S. Aranovskiy, R. Ushirobira, M. Korotina and V. Vedyakov, On preserving-excitation properties of Kreisselmeiers regressor extension scheme, \TAC, (to appear) 2022. {\tt (hal-03245139/document})


\bibitem{BER}
P. Bernard, {\em Observer Design for Nonlinear Systems}, Springer, Switzerland, 2019.

\bibitem{BERPRA}
P. Bernard and L. Praly, Convergence of gradient observer for rotor position and magnet flux estimation of permanent magnet synchronous motors, {\em Automatica}, vol. 94, pp. 88--93, 2018.

\bibitem{BERPRAtac}
P. Bernard and L. Praly, Estimation of position and resistance of a sensorless PMSM: A nonlinear Luenberger approach for a nonobservable system, \TAC, vol. 66, pp. 481--496, 2020.

\bibitem{BOBetal}
A. Bobtsov, A. Pyrkin, R. Ortega, S. Vukosavic, A.M. Stankovic and E.V. Panteley, A robust globally convergent position observer for
the permanent magnet synchronous motor, {\em Automatica}, vol. 61, pp. 47--54, 2015.

\bibitem{BOLetal}
I. Boldea, M.C. Paicu, G.-D. Andreescu, Active flux concept for motion-sensorless unified AC drives, \emph{IEEE Trans. Power Syst.}, vol. 23, pp. 2612--2618, 2008.

\bibitem{CHOetal}
J. Choi, K. Nam, A.A. Bobtsov and R. Ortega, Sensorless control of IPMSM based on regression model, {\em IEEE Trans. Power Electron.}, vol. 34, no. 9, pp. 9191--9201, 2019.

\bibitem{CHONAM}
J. Choi and K. Nam, Model-based sensorless control for IPMSM providing seamless transition to signal injection method, {\em Tech. Report}, 2019.

\bibitem{HOL}
J. Holtz, Sensorless control of induction motor drives, {\em Proc. IEEE}, vol. 90, pp. 1359--1394, 2002.


\bibitem{KRE}
G. Kreisselmeier, Adaptive observers with exponential rate of convergence, \TAC, vol. 22, pp. 2-8, 1977.

\bibitem{LEEetal} 
J. Lee, J. Hong, K. Nam, R. Ortega, A. Astolfi and L. Praly, Sensorless control of surface-mount permanent magnet synchronous motors based
on a nonlinear observer, {\em IEEE Trans. Power Electron.},  vol. 25, no. 2, pp. 290--297, 2010.

\bibitem{Li}
P.Y. Li, Self-sensing dual push-pull solenoids using a finite dimension flux-observer, \ACC, pp. 590--595, Denver, USA, July 1-3, 2020.

\bibitem{LIO}
P.M. Lion, Rapid identification of linear and nonlinear systems, {\em AIAA Journal}, vol. 5, pp. 1835-1842, 1967.

\bibitem{MALPRA}
J. Malaiz\'e and L. Praly, Robust position estimation for permanent magnet
synchronous electrical machines with salient poles, Archive ouverte HAL, 2020. (\texttt{hal-02568844}).


\bibitem{MALetal}
J. Malaiz\'e, L. Praly and N. Henwood, Globally convergent nonlinear observer for the sensorless control of surface-mount permanent magnet synchronous machines, {\em IEEE Conf. Decis. Control}, pp. 5900--5905, 2012.


\bibitem{NAM}
K. Nam, {\em AC Motor Control and Electric Vehicle Applications}, CRC Press, 2017.


\bibitem{ORTetalTCST}
R. Ortega, L. Praly, A. Astolfi, J. Lee and K. Nam, Estimation of rotor position and speed of permanent magnet synchronous motors with guaranteed stability, vol. 19, pp. 601--614, 2010.


\bibitem{ORTetal}
R. Ortega, B. Yi, S. Vukosavic, K. Nam and J. Choi, A globally exponentially stable position observer for interior permanent magnet synchronous motors, \AUT, vol. 125, Art. no. 109424, 2021.

\bibitem{ORTetalARC}
R. Ortega, V. Nikiforov and D. Gerasimov, On modified parameter estimators for identification and adaptive control: A unified framework and some new schemes, {\em Annu. Rev. Control}, vol. 50, pp. 278--293, 2020.

\bibitem{ORTetalTAC}
R. Ortega, S. Aranovskiy, A. Pyrkin, A Astolfi and A. Bobtsov, New results on parameter estimation via dynamic regressor extension and mixing: Continuous and discrete-time cases, \TAC, vol. 66, no. 5, pp. 2265-2272, 2021.

\bibitem{POUPRAORT} 
F. Poulain, L. Praly and R. Ortega, An observer for permanent magnet synchronous motors with application to sensorless control, {\em IEEE Conf. Decis. Control}, Cancun, Mexico, 9-11 December, 2008.

\bibitem{VERetal}
C.M. Verrelli, E. Carfagna, M. Frigieri, A.S. Crinto and E. Lorenzani, A new Bernard-Praly-like observer for sensorless IPMSMs, \AUT, vol. 140, Art. no. 110266, 2022.

\bibitem{YIORT}
B. Yi and R. Ortega, Conditions for convergence of dynamic regressor extension and mixing parameter estimators using LTI filters, \TAC, 2022. {\tt (DOI: 10.1109/TAC.2022.3149964)}

\bibitem{YIetal}
B. Yi, S.N. Vukosavi{\'c}, R. Ortega, A.M. Stankovi{\'c} and W. Zhang, A new signal injection-based method for estimation of position in interior permanent magnet synchronous motor, vol. 13, no. 9, pp. 1865--1874, {\em IET Power Electron.}, 2020.

\bibitem{IFAC}
B. Yi, R. Ortega, J. Choi and K. Nam, A novel globally exponentially stable observer for sensorless control of the IPMSM via Kreisselmeier's extension, submitted to {\em IFAC World Congress}, 2023.


\end{thebibliography}
\end{document}